\documentclass{article}
\usepackage{authblk}
\usepackage{graphicx}
\usepackage[utf8]{inputenc}
\usepackage{xcolor}

\usepackage{amsbsy}
\usepackage{amsmath}
\usepackage{amsfonts}
\usepackage{amssymb}
\usepackage{amsthm}
\usepackage{graphicx}
\usepackage{csquotes}
\usepackage{enumitem}
\usepackage{thmtools, thm-restate}
\usepackage{subcaption}
\usepackage{hyperref}

\graphicspath{{figs/}}

\newtheorem{observation}{Observation}{\bfseries}{\itshape}
\newtheorem{theorem}{Theorem}{\bfseries}{\itshape}
\newtheorem{lemma}{Lemma}{\bfseries}{\itshape}
\newtheorem{corollary}{Corollary}{\bfseries}{\itshape}

\author[1]{Sarah Carmesin}
\author[1]{Andr\'e Schulz}
\affil[1]{FernUniversit\"at in Hagen}
\title{Arrangements of orthogonal circles with many intersections}

\begin{document}
%


%
%

%

%
\maketitle              
\begin{abstract}
An arrangement of circles in which circles intersect only
in angles of $\pi/2$ is called an 
\emph{arrangement of orthogonal circles}. We show that in 
the case that no two circles are nested, the intersection
graph of such an arrangement is planar. The same result holds 
for arrangement of circles that intersect in an angle of at most $\pi/2$.

For the general case we prove that the maximal number of
edges in an intersection graph of an arrangement of orthogonal circles lies in between 
$4n - O\left(\sqrt{n}\right)$ and $\left(4+\frac{5}{11}\right)n$, for
$n$ being the number of circles. Based on the lower bound we
can also improve the bound for the number of triangles in 
arrangements of orthogonal circles to $(3 + 5/9)n-O\left(\sqrt{n}\right)$.

\end{abstract}

\section{Introduction}
A collection of $n$ circles in the plane, is called an \emph{arrangement of orthogonal
circles} if any two intersecting circles intersect orthogonally. Here, 
we call an intersection orthogonal, if the tangents at the intersection
point form an angle of $\pi/2$. By definition circles cannot touch 
in an arrangement of orthogonal
circles.

A natural object that arises from  an arrangement of orthogonal
circles is its intersection graph.
A graph $G$ is a \emph{(geometric) intersection graph} if its vertices can be realized
by a set of geometric objects, such that two objects intersect if and only if their 
corresponding vertices form an edge in $G$. Thus, for an arrangement of orthogonal
circles $\mathcal{A}$ we define its intersection graph $G(\mathcal{A})$ as the graph,
whose vertices correspond to the circles in $\mathcal{A}$ and two vertices are adjacent, 
if and only if the associated circles intersect in $\mathcal{A}$. The graph $G(\mathcal{A})$
is called an \emph{orthogonal circle intersection graph}.

Arrangements of orthogonal circles and their intersection graphs were recently introduced
by Chaplick et al.~\cite{cfkw19}. Here it was shown that the intersection graph of $n$ circles contains
at most $7n$ edges. Furthermore, it is {\sf NP}-hard
to test whether a graph is an orthogonal unit circle intersection graph.
Chaplick et al. also provide bounds for the maximal number of digonal, triangular and 
quadrilateral cells in arrangements of orthogonal circles.

\paragraph{Related work.}
General (non-orthogonal) arrangements of circles or disks have been studied extensively before.
Giving a complete overview over the results in this field is out of scope 
for this article. We will hence only mention a few selected results.
For the special case where all circles have the same radius the intersection
graphs are known as \emph{unit disk graphs}.
For general arrangements of circles or balls the recognition problems for the corresponding intersection graphs are usually hard (for example for unit disk graphs~\cite{bk98}).
We refer the reader to  the survey of Hlinen{\'{y}} and Kratochv{\'{\i}}l~\cite{hk01} for more information.  Other work focused on bounding the number of small faces in arrangements of circles~\cite{alps01}
or about the circleability of topologically described arrangements~\cite{km14}~\cite{fs18}.

Note that we can have general circle arrangements in which all circles pairwise intersect. Thus,
the density of the intersection graph can be $\Theta(n^2)$, although many graphs
are not intersection graphs of circle arrangements~\cite{mm14} (for example
every graph containing $K_{3,3}$ as a subgraph~\cite{hk01}). Hence, asking for
the maximum density for intersection graphs in this setting is not an interesting question.

If the circles are allowed to only intersect pairwise in one point, then the intersection
graph is called a \emph{contact graph} and the corresponding arrangement is a circle
packing. Due to the famous Andreev--Koebe--Thurston circle packing
theorem~\cite{a-cpls-70,Koebe36} the disk contact graphs coincide with the planar graphs. 
One direction of the circle packing theorem is obvious, a planar straight-line
drawing of the contact graph can be derived by placing the vertices at the disk centers. A 
related result is due to Alon et al.~\cite{alps01}. A \emph{lune}
is a digonal cell in an arrangement of circles. If we restrict the intersection
graph of the (general) circle arrangement to intersections that are formed by lunes 
(we call
this the \emph{lune-graph}) 
then also in this setting we can obtain a planar straight-line drawing by placing the vertices at the 
circle centers.

Every arrangement of orthogonal circles with the same radius can be turned into a unit circle
packing by shrinking the circle size by a factor of $\sqrt{2}/2$, but there are unit disk
contact graphs that are not intersection graphs of an arrangement of orthogonal circles~\cite{cfkw19}.

A well established quality criteria for drawing graphs is to avoid crossings. However,
crossings with large angles are considered less problematic~\cite{hhe08}. For this
reason graphs that can be drawn with right-angle crossings (known as RAC-drawings)
are considered an interesting class from a graph drawing perspective. 
It was shown that graphs that have straight-line RAC-drawings have at most $4n-10$ 
edges~\cite{del11}, for $n$ being the number of vertices.
Recently this approach was carried 
over to drawings with circular arcs that can intersect in right angles only. 
Chaplick et al. showed that graphs that have circular arc RAC-drawings can have at most $14n-12$ edges
and there are such graphs with $4.5n-O(\sqrt{n})$ edges~\cite{cfkw20}.  

Orthogonal circle arrangements can also be seen as circular arc drawings (of 4-regular graphs) with
perfect angular resolution. Such drawings are known as Lombardi drawings and have been
studied deeply~\cite{degkn12,e14,KKL+19}.
 
\paragraph{Results.} We prove bounds for the maximal number of edges in an 
intersection graph of an arrangement of $n$ orthogonal circles. 
We 
show an upper bound of $\left(4+\frac{5}{11}\right)n$
and present a lower bound of $4n - O\left(\sqrt{n}\right)$. As a crucial intermediate result we show
that in the case of arrangements without nested circles, the intersection graph is planar. In particular,
(in a similar vein to disk contact graphs and lune graphs) 
we obtain a planar straight-line drawing by placing the vertices at the centers of the corresponding 
circles. 
As an immediate consequence we get that for
arrangements of nonnested orthogonal circles the intersection graph has at most $3n-6$ edges.
We can refine the analysis to improve this bound to $3n-8$.
This bound is tight, since we can show a matching lower bound.
Our lower bound constructions can be slightly modified to also improve the bounds for 
the maximal number of triangular cells in arrangements of orthogonal circles
to $(3 + 5/9)n-O\left(\sqrt{n}\right)$.

\paragraph{Organization.}
We first prove in Section~\ref{sec:nonnest} that the orthogonal circle intersection graphs 
are planar in the nonnested case. In~Section~\ref{sec:uppernest} we extend our ideas to 
general circle arrangements and prove the upper bound. In 
Section~\ref{sec:lowernest} we discuss\
lower bound constructions.

\section{Bounds for nonnested arrangements}\label{sec:uppernonnest}
\label{sec:nonnest}

For an arrangement of orthogonal circles we call the straight-line drawing of its
intersection graph that is obtained by placing the vertices on the corresponding 
circle
centers the \emph{embedded intersection graph}. \autoref{fig:exNonested} 
depicts such an arrangement and its embedded intersection graph.
In this section we prove that the embedded intersection graph is noncrossing. 

   \begin{figure} [htpb]   
\begin{center}
    	  \includegraphics[scale=0.7]{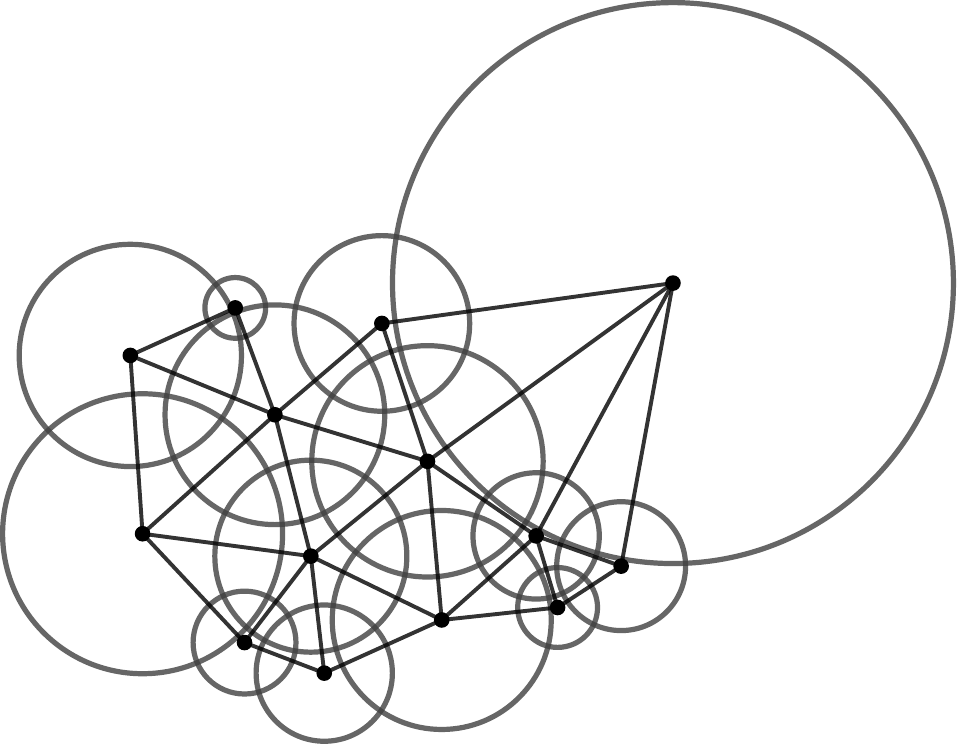}
   	  \caption{A nonnested circle arrangement and its embedded 
intersection graph.
}\label{fig:exNonested}
\end{center}
   \end{figure}

We start with properties of arrangements of two or three nonnested orthogonal circles. The first observation is a simple application of Pythagoras' theorem.

\begin{observation}
\label{distance}
  Let $A$ and $B$ be two circles with centers $C_A$ and $C_B$ and radii $r_A$ and $r_B$, respectively. Then $A$ and $B$ are orthogonal if and only if $|C_AC_B|^2=r_A^2+r_B^2$.
\end{observation}

\begin{lemma}\label{lem1}
 In an arrangement of nonnested orthogonal circles, the center of a circle $A$ is not contained inside a circle other than $A$.
\end{lemma}

\begin{proof}
 Let $A$ and $B$ be two nonnested circles with centers $C_A$ and $C_B$ and radii $r_A$ and $r_B$, respectively. Assume that $C_A$ lies inside $B$.  Obviously, $A$ and $B$ intersect, since otherwise the circles are nested.
 Since $A$ and $B$ intersect orthogonally it holds that $|C_AC_B|^2=r_A^2+r_B^2$. Further, since $C_A$ is in $B$, we have $|C_AC_B|<r_B$ and thus $|C_AC_B|^2<r_B^2$.  We get that  $r_A^2+r_B^2 < r_B^2$, which is a contradiction for $r_A,r_B \in \mathbb R$.
 \end{proof}

\begin{lemma}\label{lem2}
 In an arrangement of nonnested orthogonal circles, for every pair of circles $A$ and $B$ and every point $p$ on $A$ it holds that $B$ intersects the line segment $C_Ap$ in at most one point.
\end{lemma}

\begin{proof}
 Let $A$ and $B$ be two circles with centers $C_A$ and $C_B$ and radii $r_A$ and 
$r_B$, respectively. Assume that there exist a point $p$ on $A$ such that $B$ 
intersects $C_Ap$ twice. We call these intersection points $q$ and $s$ with 
$|C_Aq|<|C_As|<r_A$ and denote the midpoint between $q$ and $s$ with $t$. 
 By \autoref{lem1} $C_B$ lies outside of $A$. So, for the circle $B$ to have a point inside of  $A$,  $B$ has to intersect $A$ in some point $u$ (see \autoref{fig:lem2}). 
 Since the circles intersect orthogonally, $C_AuC_B$ is a right triangle. 
Further, since $qs$ is a chord of the circle $B$, the triangle $sqC_B$ is 
isosceles and its height is $C_Bt$. Thus, we have a right angle at $t$ between 
$C_Ap$ and $C_Bt$ and  $C_AtC_B$ is  a right triangle.

   \begin{figure} [htpb]   
\begin{center}
    	  \includegraphics[scale=0.6]{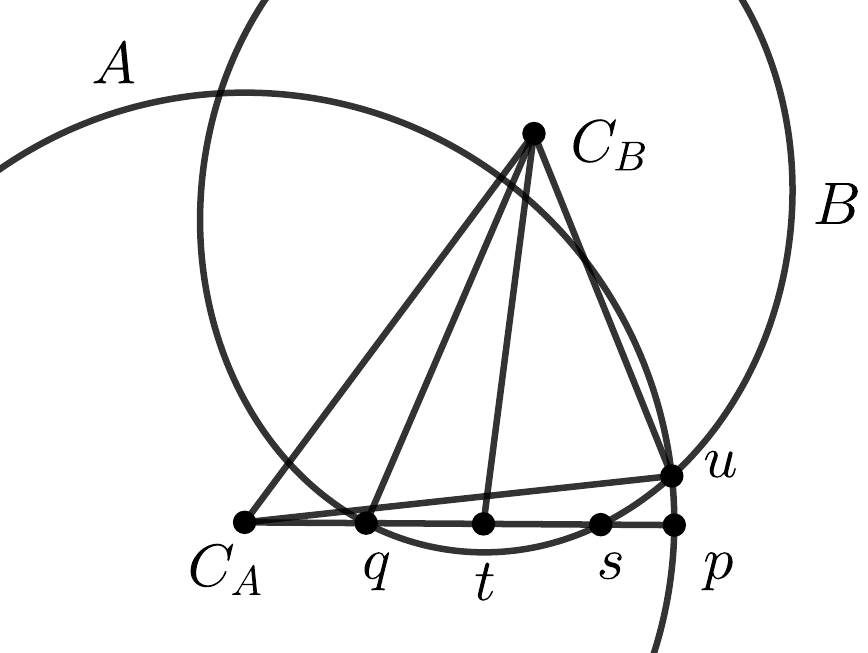}
   	  \caption{Illustration of the construction in the proof of \autoref{lem2}.
}\label{fig:lem2}
\end{center}
   \end{figure}

Since the right triangles $C_AtC_B$ and $C_AuC_B$ share the same hypotenuse 
$C_AC_B$  we get by 
Thales' theorem that 
$t$ and $u$ have to be on the circle with diameter $C_AC_B$.
We know that $|C_Bu|=r_B$ and $|C_Bt|<|C_Bq|=r_B$. Hence, $t$ lies closer to $C_B$ than $u$. It follows that $u$ is closer to $C_A$ than $t$, thus $C_At>C_Au=r_A$. This implies that $t$ is outside of $A$, which is a contradiction.
\end{proof}

\begin{lemma}\label{lem3}
 In an arrangement of nonnested orthogonal circles, for every intersecting pair of circles $A$ and $B$ there is no third circle 
 that intersects the line segment between the centers of $A$ and $B$.
\end{lemma}

\begin{proof}
Let $A$, $B$ and $D$ be three circles with centers $C_A$, $C_B$ and $C_D$ and radii $r_A$, $r_B$ and $r_D$, respectively. The circles $A$ and $B$ intersect. Assume for
a contradiction that the circle $D$ intersects the line between $C_A$ and $C_B$.

If the circle $D$ intersects the line segment between $C_A$ and $C_B$ just once, either $C_A$  or $C_B$ would be inside $D$, a contradiction to \autoref{lem1}. Thus, $D$ has to intersect the line segment $C_AC_B$ twice.
We denote these intersection points by $q$ and $s$ with $|C_Aq|<|C_As|<|C_AC_B|$ and midpoint between $q$ and $s$ by $t$. 
Due to \autoref{lem2} $q$ and $s$ cannot lie in the same circle, so one lies 
in $A$ and the other in $B$. Thus, $D$ intersects both $A$ and $B$.
By \autoref{lem1} the center $C_D$ of  $D$ has to be outside of the circles $A$ and $B$. Thus, for the circle $D$ to have a point in the inside of the circles $A$ and $B$, the circle $D$ has to intersect the circles $A$  (in some point $u_A$) and $B$ (in some point $u_B$). The situation is depicted in \autoref{fig:lem3}.

\begin{figure} [htpb]   
\begin{center}
    	  \includegraphics[scale=0.66]{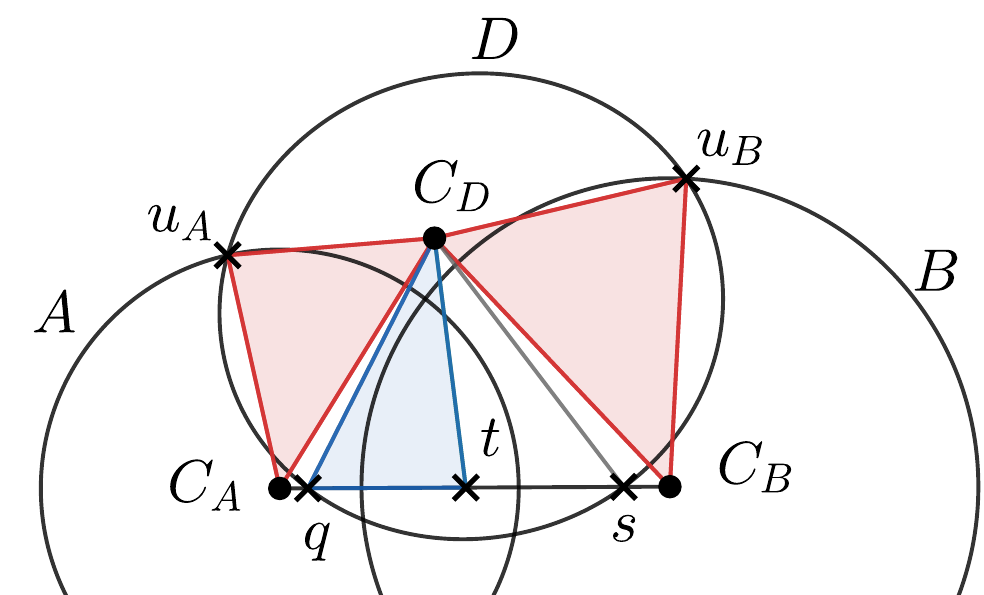}
   	  \caption{Illustration of the construction in the proof of \autoref{lem3}.
}\label{fig:lem3}
\end{center}
\end{figure}

Since all circles intersect orthogonally we have right angles at $u_A$ between $C_Au_A$ and $C_Du_A$ and  at $u_B$ between $C_Bu_B$ and $C_Du_B$. Also, since $sq$ is a chord of the circle $D$, the triangle $sqC_D$ is isosceles and its height is $C_Dt$. Thus, we have a right angle at $t$ between $C_AC_B$ and $C_Dt$.
This gives us five right triangles $C_AC_Du_A$ and $C_BC_Du_B$ (red), $qC_Dt$ 
(blue), $C_BC_Dt$ and $C_AC_Dt$. We obtain

\begin{gather*}
    |C_AC_D|^2 =r_A^2+r_D^2, \quad 
    |C_BC_D|^2 =r_B^2+r_D^2, \quad |C_BC_D|^2 = |C_Bt|^2+|C_Dt|^2 \\
        |C_AC_D|^2 = |C_At|^2+|C_Dt|^2, \qquad
r_D^2 =\left(\frac{|qs|}{2}\right)^2+|C_Dt|^2.  
  \end{gather*}

Combining these equations we get 

\begin{align*}
	|C_At|^2 & = |C_At|^2 + |C_Dt|^2 - |C_Dt|^2 = |C_AC_D|^2 - |C_Dt|^2 =r_A^2+r_D^2 - |C_Dt|^2\\
	& = r_A^2 +\left(\frac{|qs|}{2}\right)^2. 
\end{align*}

It follows that $|C_At| > r_A$. By a symmetric argument we see also that $|C_Bt| > r_B$.
We get   $|C_At|+|C_Bt| > r_A + r_B$, which is a contradiction.
%
%
 %
\end{proof}

%
%
%
%
%
%


We can now combine our observations to prove the following result.

\begin{theorem}\label{thm:main}
 The embedded intersection graph of an arrangement of nonnested orthogonal circles is noncrossing.
\end{theorem}

\begin{proof}
Suppose for contradiction four circles $A,B,D,E$ with centers $C_A,C_B,C_D$ and $C_E$ that are arranged in such way that their embedded intersection graph has two edges $C_AC_B$ and $C_DC_E$ that cross in the point $h$. 
This means we have two pairs of intersecting circles $A,B$ and $D,E$.
Note that  $C_AC_B$ is contained in the union of $A$ and $B$. 
Hence, $h$ has to lie in at least one of the circles $A$ or $B$. By the same reasoning
$h$ also has to lie in at least one of the circles $D$ or $E$. Without loss of generality
we can assume that $h$ lies in $D$.
By \autoref{lem1} the circle $D$ cannot enclose $C_AC_B$ completely, thus it has to intersect the line segment $C_AC_B$. This, however, contradicts \autoref{lem3}.
\end{proof}

By \autoref{thm:main} the intersection
graph is planar and we can further show that the boundary face  of the embedded intersection 
graph is at least a pentagon if we have five or more circles (Lemma~\ref{lem:pentagon} in Appendix~\ref{app:omitted2}). Applying Euler's formula yields the following result.

\begin{corollary}
\label{cor:1}
The intersection graph of an arrangement of nonnested orthogonal circles has at most $3n-8$ edges
for $n\ge 5$. 
\end{corollary}


In Section~\ref{sec:lowernest}  we show that the bound of $3n-8$ in \autoref{cor:1} is tight.
In Appendix~\ref{app:acute}
we show that Theorem~\ref{thm:main} also holds when all circles intersect in a (not necessarily
identical) angle of at most $\pi/2$.

\section{Bounds for general orthogonal arrangements}
\label{sec:uppernest}

In this section we prove an upper bound of $\left(4+\frac{5}{11}\right)n$ edges for intersection graphs of orthogonal circle arrangements with nested circles.
We first discuss the general approach and introduce necessary terminology before 
continuing with details and proofs.

For every circle $C$ in an arrangement $\mathcal A$ we define its \emph{depth} $t(C)$ as 
the maximum cardinality of a set of pairwise nested circles in $\mathcal A$ that are properly 
contained in $C$. 
A circle with depth $0$, i.e., it contains no circles properly, is referred to as \emph{shallow}
otherwise as \emph{deep}.

As a first step we show that in every arrangement we can find a circle with depth at most $1$ that is 
orthogonal to at most seven deep circles (\autoref{7er}). We select one circle with
this property and name it the 
\emph{red circle}. We then look at the circles properly contained in the red circle. We call
these circles \emph{black circles}; see \autoref{fig:8n}. 
The key observation is that we can delete the set of black circles from the arrangement and 
 by doing so we only lose few edges from the intersection graph, i.e.  at most $(4+ 5/11) \cdot n$ 
for $n$ black circles. To obtain this bound we distinguish between intersections between the black circles and intersections between a black circle and a circle that intersects a black and the red
circle (such circles are called \emph{green circles}).
To make our analysis work we have to partition the black circles further. 
If a black circle center lies on the boundary of the embedded intersection graph induced by
the vertices of the black circles we call the corresponding circle 
\emph{boundary black circle}, otherwise \emph{inner black circle}.


 \begin{figure} [htpb]   
\begin{center}
    	  \includegraphics[height=4cm]{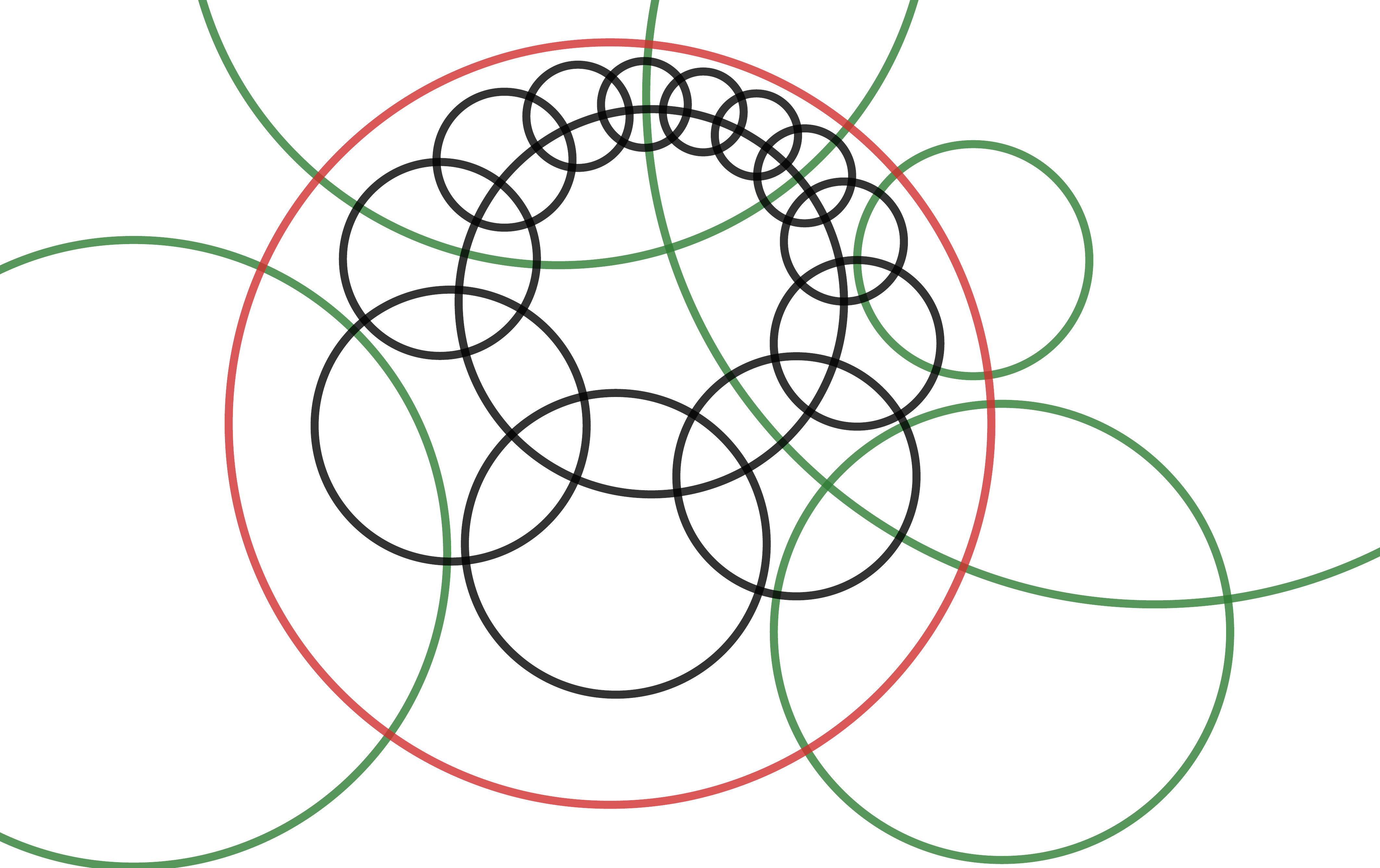}
   	  \caption{Illustration of the red, black and green circles. This arrangement has
   	  only one inner black circle.}\label{fig:8n}
\end{center}
\end{figure}

We color edges in the intersection graph according to the color of the corresponding 
circles as follows: An intersection between a black and a green circle yields a \emph{green edge} 
and an intersection between two black circles yields a \emph{black edge}.
 If there are $n$ black circles and $b$ of those are boundary black circles, 
 then we have at most $3n-b-3$ black edges as a consequence of Euler's formula and Theorem~\ref{thm:main}. 
 We will prove that
each black circle can be orthogonal to at most two green circles (\autoref{twoCircles}). 
%
However, 
 the inner black circles can only be intersected by green circles with depth at least $1$ (\autoref{inner}). We can chose
  the red circle so that there are at most seven deep green circles. The intersection graph of these
   seven circles has at most eight edges (\autoref{noC34} and \autoref{8edges}). We exploit this 
 fact to show that only eight inner circles can be orthogonal to two green circles (\autoref{7Points}).
As a final observation we show that if there are at most $11$ black circles in the red circle, there are at most $3$ inner black circles that intersect two of the green circles.
We can then combine our findings to prove that we can always find a set of $n$ black circles 
that intersects at most $(4 + 5/11) \;n$ circles.

We are now continuing with the proofs and details. We begin by stating a few properties of arrangements of orthogonal circles. The first lemma was proven by Chaplick~et~al.~\cite{cfkw19}. 

\begin{lemma}[\cite{cfkw19}]\label{K4C4}
 No orthogonal circle intersection graph contains a $K_4$ or an induced $C_4$.
\end{lemma}


\begin{lemma}\label{twoCircles}
 Let $A$ and $B$ be two nested circles. There are at most two circles that intersect both $A$ and $B$  orthogonally.
\end{lemma}

\begin{proof}
 Suppose that there are two nested circles $A$ and $B$ ($B$ lies inside $A$) that both
intersect at least three circles $D$, $E$ and $F$  orthogonally.
 Consider the intersection graph of $A$, $B$, $D$, $E$ and $F$.
If the circles $D$, $E$ and $F$ are pairwise orthogonal to each other the vertices corresponding of $A$, $D$, $E$ and $F$ form a $K_4$, a contradiction due to \autoref{K4C4}.
However, if two of the circles $D$, $E$ and $F$ are not orthogonal to each other their corresponding vertices together with $B$ and $A$ induce a $C_4$, which yields a contradiction due to \autoref{K4C4}.

%
%
Thus, at most two circles that intersect both $A$ and $B$ orthogonally.
\end{proof}

\begin{lemma}\label{intersection} 
If a circle $C$ intersects the circles $A$ and $B$ orthogonally, then one of the following holds:
%
(i) $A$ and $B$ do not intersect, or
(ii) $A$ and $B$ are orthogonal and $C$ contains precisely one of the intersection points of $A$ and $B$.

\end{lemma}

\begin{proof}
We prove that if (i) does not hold, then (ii) holds. So assume $A$ and $B$ intersect.
We apply a Möbius transformation that maps $A$ to a straight line. Note that such a transformation
is conformal and thus maintains the angles; see \autoref{fig:lemma6}. The centers of $B$ and $C$ will then have to lie on $A$. 
Clearly, if $C$ contains both points of $A\cap B$ then it also has to contain $B$, but 
since $B$ intersects $C$, we have a contradiction.
Also, if $C$ does not contain any point of $A\cap B$, then it has to be either contained in $B$ or is to the left or right of $B$ along $A$,
but since $B$ intersects $C$, we have again a contradiction.
%
%
\end{proof}

\begin{figure}
	\centering
\includegraphics[scale=.7]{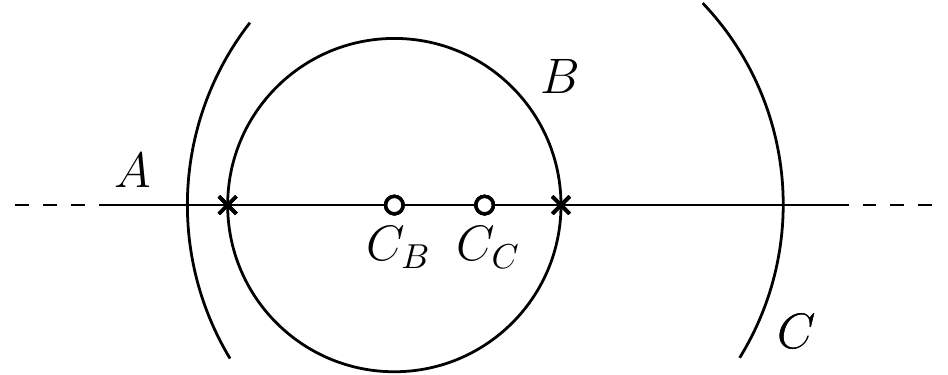}
	\caption{Situation in the proof of Lemma~\ref{intersection} when $C$ contains $A\cap B$. $C_B$ ($C_C$) is the center of $B$($C$).}
	\label{fig:lemma6}
\end{figure}

\begin{lemma}\label{justOne} 
In an arrangement of orthogonal circles let $A$ and $B$ be two circles that intersect.
All circles that are orthogonal to $A$ and $B$ that contain the same intersection point of $A$ and $B$ are nested.
\end{lemma}
\begin{proof}
Assume that there are two nonnested circles $C$ and $D$ that both contain the same intersection point $u$ of $A$ and $B$. Since $C$ and $D$ contain $u$ but are not nested, they must intersect each other. Both also intersect $A$ and $B$. This means the intersection graph of the four circles is a $K_4$. This contradicts \autoref{K4C4}.
\end{proof}

%
%
%
%
 


The following lemma is again taken from Chaplick~et~al.~\cite[Lemma~5]{cfkw19}. 
The \enquote{Moreover}-part is not explicitly  written down, but 
it is apparent from the construction given in its proof. 

\begin{lemma}[\cite{cfkw19}] \label{chaplick7erStrong} 
 Every arrangement of orthogonal circles has a circle that is orthogonal to at most seven other circles.
 Moreover, this circle is a shallow circle. 
\end{lemma}

We can deduce a similar lemma for deep circles.

\begin{lemma} \label{7er}
 Every arrangement of orthogonal circles with nested circles has a circle $C$ with depth $t(C)=1$ that is orthogonal to at most seven other circles with depth at least $1$.
\end{lemma}

\begin{proof}
Let $\mathcal A$ be an arrangement of orthogonal circles. By deleting all shallow circles we obtain the arrangement $\mathcal A'$. According to \autoref{chaplick7erStrong} we can find a shallow circle $C$ in $\mathcal A'$ that is orthogonal to at most seven other circles. 
Since $C$ is shallow in $\mathcal A'$ it has depth $t(C)=1$ in the arrangement $\mathcal A$. 
\end{proof}

In the following we select any circle that meets the requirements of Lemma~\ref{7er} and refer to it as the red circle. We remind the reader that we call the circles contained in the red circle
the black circles.
 
 \begin{lemma}\label{4mi}
The set of black circles $S_B$ inside a red circle $C$ corresponds to a vertex set $V_B$ in the intersection graph incident to 
no more than $4n_B+i-3$ edges,
for $n_B=|S_B|$ and $i$ being the number of inner black circles in $S_B$, that each are orthogonal to two circles not in $S_B$.
\end{lemma}
\begin{proof}

Let $C$ be the red circle. 
%
%
We count the edges incident to $V_B$. Edges with two endpoints in $V_B$ are black edges, edges with
one endpoint in $V_B$ are green edges. 
We denote the number of boundary black circles by $b$.
According to \autoref{thm:main} the intersection graph of 
the arrangement restricted to the $S_B$ is planar. Moreover this planar graph has $b$ vertices
on its outer face. Thus, by Euler's formula we have at most than $3n_B-b-3$ black edges. 

We now count the green edges. Every circle $D \notin S_B$ that intersects a circle in $S_B$ has to intersect $C$ as well. 
According to \autoref{twoCircles}, each of the $n_B$ black circles 
is orthogonal to at most two green circles. 
By our assumption $n_B-b-i$ black inner circles intersect at most one green circle. Thus, we have 
at most $2n_B-(n_B-b-i)=n_B+b+i$ green edges. Adding the $3n_B-b-3$ black edges yields the upper bound of 
$4n_B+i-3$ as stated in the lemma.
%
%
%
\end{proof}
 

\begin{lemma} \label{inner} 
Every green circle intersecting an inner black circle is a deep circle.
\end{lemma}
\begin{proof}
Let $D$ be the red circle and $S_D$ be the set of black circles.
Suppose for a contradiction that there is a shallow green circle $E$ with center $C_E$ that intersects an inner circle $F \in S_D$ with center $C_F$. 
Note that in this case $E$ also has to intersect $D$. 
By \autoref{lem1} $C_E$, is therefore outside of $D$; see \autoref{fig:inner}.


Let $\mathcal A$ be the arrangement consisting of the circles in $S_D$ and $E$. All circles in $S_D$ and the circle $E$ have depth $0$ so the arrangement $\mathcal A$ is nonnested. According to \autoref{thm:main} the embedded intersection graph  $G(\mathcal A)$ is noncrossing.

Let $\mathcal A'$ be the arrangement consisting only of the circles in $S_D$. Again, all the circles are shallow, thus the arrangement is nonnested and the intersection graph $G(\mathcal A')$ noncrossing.

Since $C_E$ is outside $D$ it lies in the outer face of $G(\mathcal A')$. On the other
hand~$F$ is an inner circle, so its corresponding vertex 
is not on the boundary of $G(\mathcal A')$. The straight-line edge between $C_E$ and $C_F$ must intersect an edge on the boundary of the outer face of $G(\mathcal A')$. This yields a crossing and thus a contradiction. 

\begin{figure}[htpb]   
\begin{center}
    	  \includegraphics[scale=0.75]{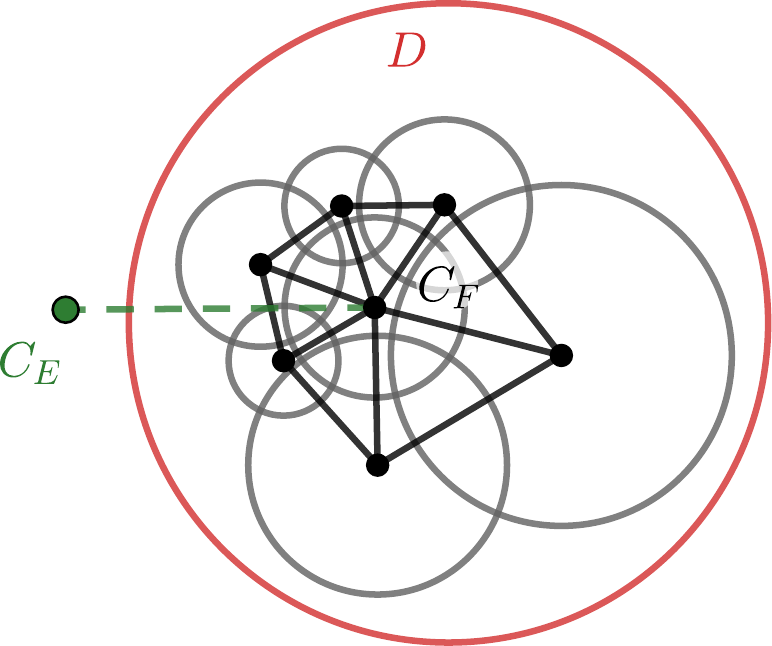}
   	  \caption{Illustration of the arrangement in \autoref{inner}.}\label{fig:inner}
\end{center}
\end{figure}

\end{proof}
 
By \autoref{chaplick7erStrong} a red circle $C$ intersects at most $7$ deep 
circles.
We now take a look at the possible intersections of the seven deep circles. We start with the
following observation.

\begin{observation}\label{noC34}
Let $I_C$ be the set of deep circles that intersect a red circle. The intersection graph of $I_C$ has 
\begin{itemize}
 \item no induced $C_4$, according to \autoref{K4C4} and
\item  no induced $C_3$, since every circle in $I_C$ is orthogonal to the red circle and according to \autoref{K4C4} there is no $K_4$ in the intersection graph of the arrangement consisting of $I_C$ and $C$.
\end{itemize} 
\end{observation}

By a case distinction we can limit the graphs that fulfil the constraints listed in \autoref{noC34}.
The proof is given in the appendix.

\begin{restatable}[]{lemma}{eightedges}\label{8edges}
 Every graph $G$ with at most seven vertices without an induced $C_3$ or $C_4$ has at most $8$ edges. 
\end{restatable}

We can now bound the number of intersection points of the circles in $I_C$.

\begin{lemma}\label{7Points}
 Let $C$ be the red circle and let $I_C$ be the set of deep circles intersecting $C$. The arrangement of circles in $I_C$ has at most sixteen intersection points of which eight are inside of $C$.
\end{lemma}
\begin{proof}
According to \autoref{8edges} the intersection graph of $I_C$ has at most eight edges. Hence, 
there are eight pairs of intersection points in the arrangement consisting of the circles in $I_C$. 
Due to \autoref{intersection} for every pair exactly one intersection point is inside of $C$. Thus, at most eight intersection points of circles in $I_C$ are inside the circle $C$.
\end{proof}

\begin{lemma}\label{4m}
 In the intersection graph of every arrangement of orthogonal circles we can find a 
 nonempty subset $V_C$ that is incident to at most $4n+5$ edges, where $n=|V_C|$.
\end{lemma}
 \begin{proof}
 Let $\mathcal A$ be an arrangement of orthogonal circles. According to \autoref{7er} we can a find a red  circle $C$ with depth $t(C)=1$ that is orthogonal to at most seven deep circles. 
  We denote the black circles by $S_C$ and set $n=|S_C|$. Further let $V_C$ denote
  the vertex set corresponding to $S_C$.
 
 We now prove that there are at most $8$ inner black circles 
 in $S_C$ that are orthogonal to two circles not in $S_C$.
 According to \autoref{inner} the inner black circles can only be intersected 
 by deep green circles.  
 If a black circle intersects two green circles, then the green circles have to intersect, otherwise the intersection graph of the black, the two green and the red circle would induce a $C_4$.
According to \autoref{intersection} a black circle that intersects two green circles contains their intersection point. \autoref{justOne} states that all circles containing the same intersection point must be nested. Since the black circles are not nested, only one black circle contains a given intersection point. By 
 \autoref{7Points} the seven deep green circles have at most eight intersection points inside $C$.  Thus, at most eight inner black circles are orthogonal to two deep green circles. 
 We now apply \autoref{4mi} with $i_2\le 8$ to obtain that $V_C$ is incident to at most $4n -3 +i_2= 4n+5$ edges.
 \end{proof}

Our goal is to apply the last lemma for bounding the density of the intersection graph. If we can 
repeatedly take out vertex sets of size $k$ with $c k$ incident edges 
(for a constant $c$), then the density
of the graph is no more than $cn$, for $n$ being the number of vertices. Unfortunately, because of the additive constant
\autoref{4m} is too weak if the subsets are small. Hence, we analyse small
sets separately to get a better bound. 
The analysis including the proofs 
can be found in the appendix. It culminates in the following statement.

\begin{restatable}[]{lemma}{four}\label{lem:4+m}
 In the intersection graph of an arrangement of orthogonal circles we can find a subset $V_C$ of $n$ vertices that has at most $\left(4+\frac{5}{11}\right) n$ edges. 
\end{restatable}

\begin{theorem}\label{lem:4+n}
 The intersection graph of an arrangement of $n$ orthogonal circles has at most $\left(4+\frac{5}{11}\right) n$  edges. 
\end{theorem}

\begin{proof}

 Assume there exist arrangements with $n$ orthogonal circles, whose intersection graphs have more than $\left(4+\frac{5}{11}\right) n$  edges. Consider a smallest such arrangement $\mathcal A$ in terms of numbers of circles and its intersection graph $G(\mathcal{A})=(V,E)$.
 By \autoref{lem:4+m} there exists a subset $S\subset V$ of $n'$ vertices 
 that is incident to at most $\left(4+\frac{5}{11}\right)n'$  edges.
We take out $S$ and all incident edges.
The new graph has $(n-n')$ vertices and more than $\left(4+\frac{5}{11}\right)(n-n')$ edges. This contradicts the assumption that  $\mathcal{A}$ is minimal.
\end{proof}


\section{Lower bounds}\label{sec:lowernest}


 \begin{figure} [htpb]   
\begin{center}
    	  \includegraphics[height=5.5cm]{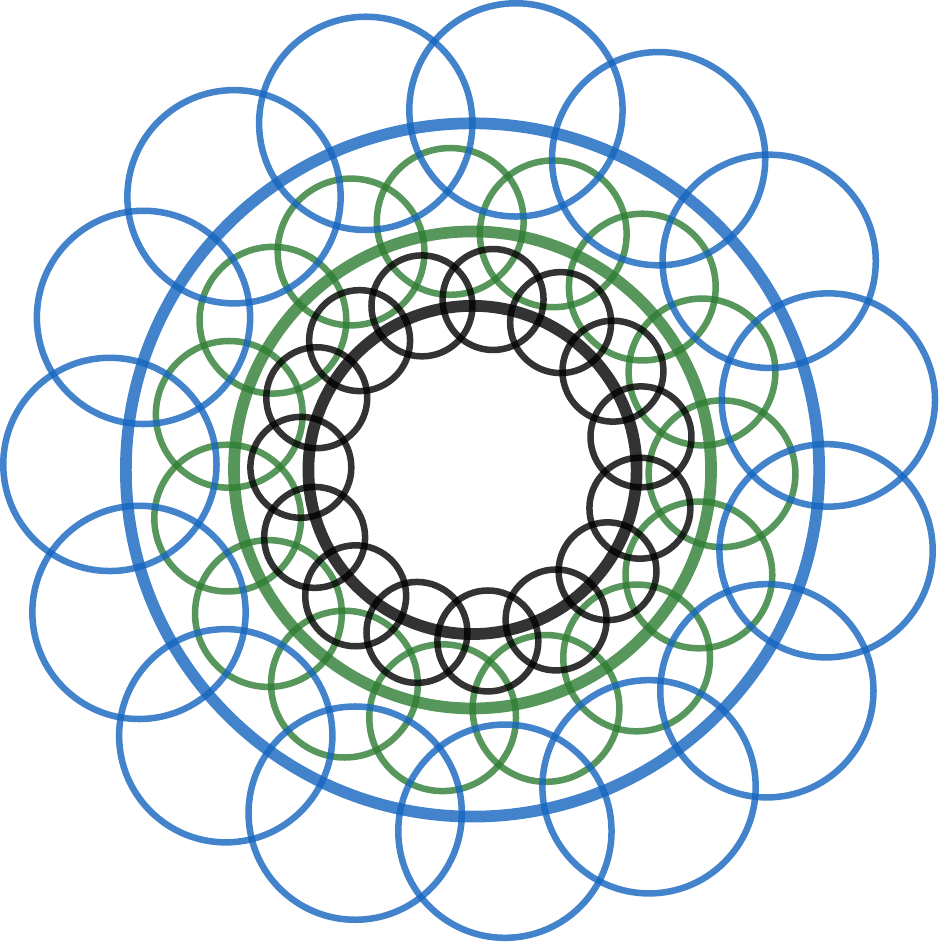}
   	  \caption{The arrangement $\mathcal B_{3,15}$. Hub circles are drawn with thick, satellite circles with thin lines. Corresponding satellite and hub circles have the same color.}\label{fig:LB8n}
\end{center}
\end{figure}

In this section we discuss lower constructions.
Our ideas are based on the arrangement $\mathcal B_{x,a}$, parametrized by two integers $a\ge 5$ and $x\ge 1$, which is constructed as follows.
We start with arranging $a$ circles with the same radius in such a way that their centers lie on a circle and two neighboring circles intersect. We call these circles the \emph{satellite circles}. We add another circle (called \emph{hub circle}) to this arrangement such that it intersects every satellite circle orthogonally. We name this arrangement a \emph{wheel of circles}. An arrangement $\mathcal B_{x,a}$ is then constructed by \enquote{nesting} $x$ wheels of circles with $a$ satellite circles each inside each other such that each satellite circle of one wheel intersects two satellite circles of the next wheel and two satellite circles of the previous wheel. The details of this 
construction (including the proof that the arrangement is orthogonal) can be found in Appendix~\ref{app:lower}.

\begin{lemma}\label{count}
 The intersection graph of $\mathcal B_{x,a}$ has $x \cdot (a+1)$ vertices and $4xa - 2a$ edges.
\end{lemma}
\begin{proof}
The arrangement consists of $x$ wheel of circles, each having $a$ satellite circles and one hub circles. Thus, the intersection graph has $x \cdot (a+1)$ vertices.
 Every vertex corresponding to a hub circle has clearly degree $a$.
Further, every vertex corresponding to a satellite circle has degree $7$, except those corresponding to a satellite circle on the inner or outermost wheel of circles, which have degree $5$.
So  the sum of the vertex degrees  is $\sum_{v\in V(G_{x,a})}\deg(v) = a x + 7 a (x-2)+ 5 a\cdot 2= 8xa-4a$. This number equals twice the number of edges, and therefore the intersection graph has
$4xa - 2a$ edges.
\end{proof}

\begin{restatable}[]{lemma}{lowerbound}\label{lem:lowerbound}
For every $n$ there is an arrangement of orthogonal circles, whose intersection graph has $n$ vertices and $4n-O\left(\sqrt{n}\right)$ edges. 
\end{restatable}
The proof of the lemma is obtained by counting the edges in the arrangement $\mathcal B_{x,a}$ with
$x$ and $a$ being $\Theta({\sqrt{n}})$. Details are given in Appendix~\ref{app:omitted}.

We now give a lower bound for nonnested orthogonal circles. The proof is given in Appendix~\ref{app:omitted}.

\begin{restatable}[]{lemma}{lowerNonNest}\label{lem:lowerNonNest}
 For every $n\ge 6$ for which $n \bmod 5 = 1$ the arrangement $\mathcal B_{((n-1)/5),5}$ with only the innermost hub circle is nonnested and its  $n$-vertex intersection graph has $3n-8$ edges.
\end{restatable}

\begin{figure} [htpb]   
\begin{center}
    	  \includegraphics[height=6cm]{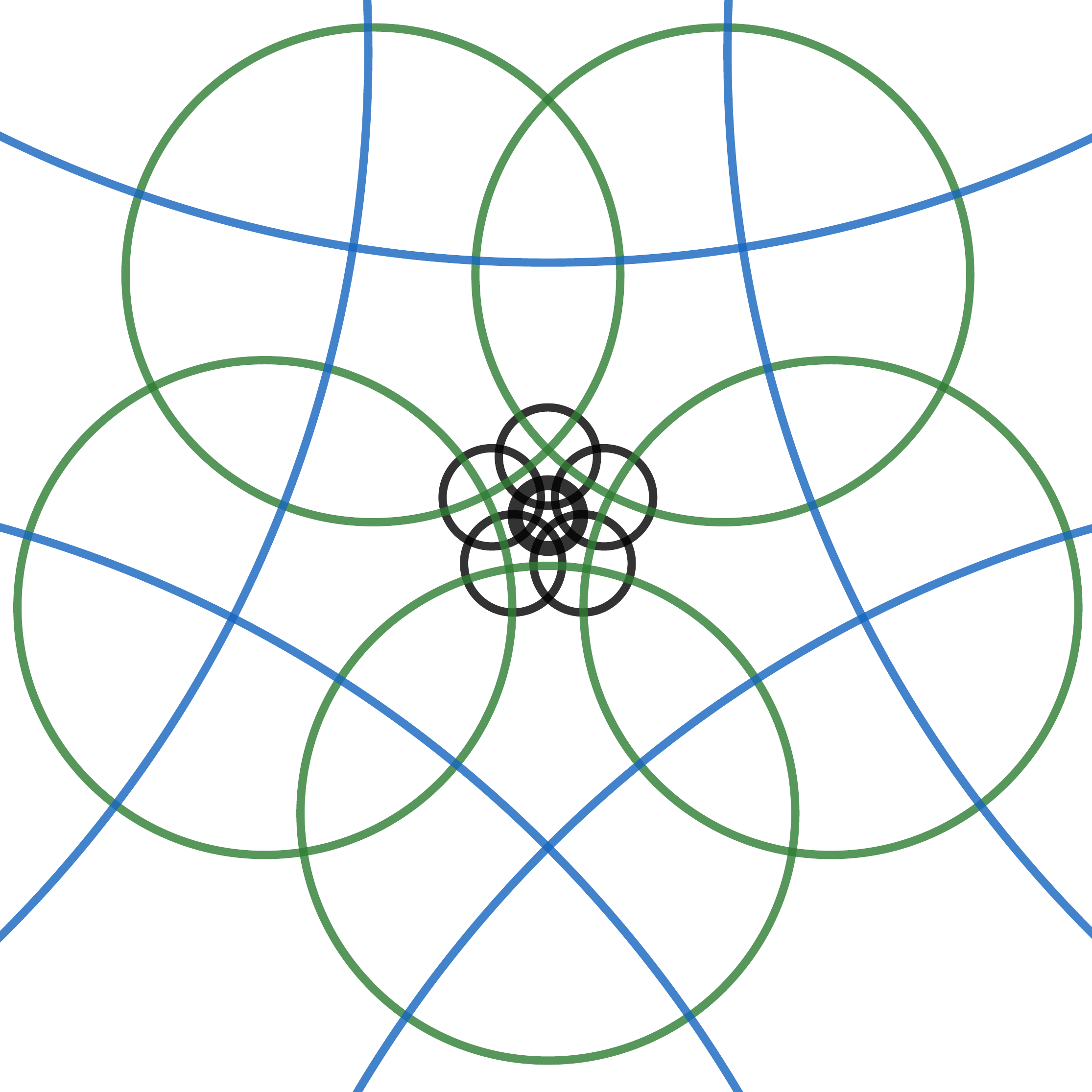}
   	  \caption{Detail of the arrangement used to prove the lower bound in the nonnested case in \autoref{lem:lowerNonNest}.}\label{fig:lowerbound5er}
\end{center}
\end{figure}

Chaplick~et~al.~\cite{cfkw19} investigated the maximal number of triangular cells in an orthogonal circle arrangement.  They proved an upper bound of $4n$ and gave a lower bound of $2n$ triangular cells, which they later improved to $3n-3$. This bound can be improved by taking the arrangement 
$\mathcal B_{x,a}$ and place a small (orthogonal) circle around every intersection point. This implies the following lemma which proof is given in Appendix~\ref{app:omitted}.

\begin{restatable}[]{lemma}{triangular}\label{lem:lowerTriangular}
  For infinitely many values of $n$ there is an arrangement of $n$ orthogonal circles with $\left(3+\frac{5}{9}\right)n-O\left(\sqrt{n}\right)$ triangular cells.
\end{restatable}

\bibliographystyle{splncs04}
\bibliography{orthogonal.bib}

\appendix

\section{Omitted material from Section~\ref{sec:uppernonnest}}
\label{app:omitted2}

In Section~\ref{sec:uppernonnest} we proved that the embedded intersection graph 
is noncrossing. As a consequence Corollary~\ref{cor:1} gave us an upper bound on 
the number of its edges. Here, we also used the fact that for orthogonal arrangements
of at least five circles, the boundary face of the embedded intersection graph is at least
a pentagon. We now prove this fact, which was omitted in the main part due
to space constraints. We start with a helpful observation.

\begin{lemma}\label{lem:three}
 In an arrangement of three pairwise orthogonal circles every point in the triangle formed by the circle centers is covered by at least one circle. 
\end{lemma}
\begin{proof}
Let $A, B, C$ be three pairwise orthogonal circles with centers $C_A,C_B,C_C$ and radii $r_A,r_B,r_C$ respectively.

Suppose for contradiction a point $p$ inside the triangle $C_A C_B C_C$ not covered by any of the three circles. So its distance to any circle center is larger than the radius of that circle. Now consider the triangle $C_AC_Bp$, as depicted in \autoref{fig:three}.
We have $|C_Ap|>r_A$ and $|C_Bp|>r_B$ and $|C_AC_B|^2=r_A^2 + r_B^2$ since $A$ and $B$ intersect orthogonally. Combining these equations yields
$$
 |C_Ap|^2+|C_Bp|^2 = r_A^2 + r_B^2> |C_AC_B|^2.
$$
It follows that the angle between $C_Ap$ and $C_Bp$ in $p$ is acute. By the same argument the angles at $p$ between $C_Ap$ and $C_Cp$, and $C_Bp$ and $C_Cp$ are also acute. 
The sum of three acute angles is less than $2\pi$, thus we have a contradiction.
\end{proof}

\begin{figure} [htpb]   
\begin{center}
    	  \includegraphics[scale=0.6]{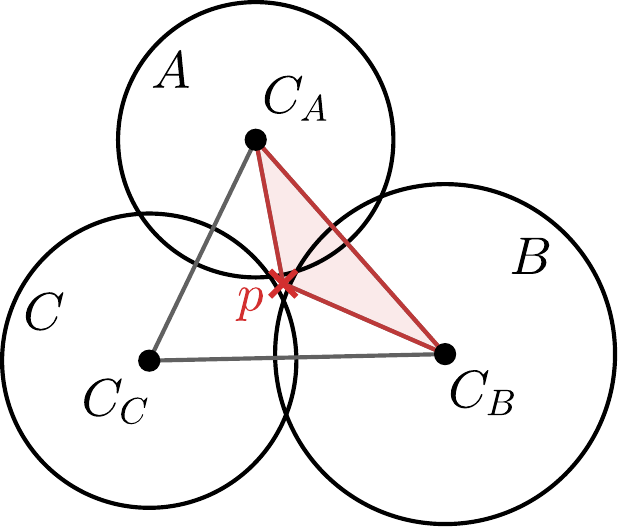}
   	  \caption{The point $p$ is inside the triangle $C_AC_BC_C$ but outside the circles $A$, $B$ and $C$}\label{fig:three}
\end{center}
\end{figure}

\begin{lemma}\label{lem:pentagon}
 The outer face of the embedded intersection graph of an arrangement of $n\geq5$ nonnested orthogonal circles 
 is adjacent to at least $5$ vertices.
\end{lemma}
\begin{proof}
 
 Assume there exists a nonnested arrangement of $n\geq5$ orthogonal circles such that the outer face of its intersection graph is adjacent to exactly three vertices. It follows that there are three pairwise orthogonal circles, whose center lie on the vertices adjacent to the outer face, as seen in \autoref{fig:5a}. According to \autoref{lem:three} three pairwise orthogonal circles cover the whole triangle between their centers, thus according to \autoref{lem1} there cannot be another circle center within that triangle, a contradiction to $n\geq5$. 
 
 Assume there exists a nonnested arrangement of $n\geq5$ orthogonal circles such that the outer face of its intersection graph is adjacent to exactly four vertices. By \autoref{K4C4} these vertices cannot induce a $C_4$ or a $K_4$, so the induced graph is a $K_4$ minus one edge as shown in \autoref{fig:5b}. By \autoref{lem:three} there can be no further circle inside each of those triangles (see discussion in the previous paragraph). Thus, the arrangement consists of at most $4$ circles. 
 
 \begin{figure}[htbp]
\captionsetup{singlelinecheck=on,justification=raggedright}
    \centering
    \begin{subfigure}[t]{0.47\textwidth}
        \centering
        \includegraphics[scale=0.8]{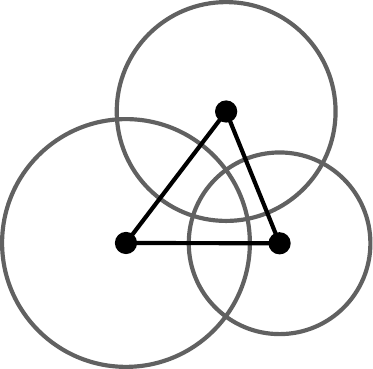} 
        \caption{An intersection graph and its circles where the outer face is adjacent to three vertices.}\label{fig:5a}
    \end{subfigure}\hfill
    \begin{subfigure}[t]{0.47\textwidth}
        \centering
        \includegraphics[scale=0.8]{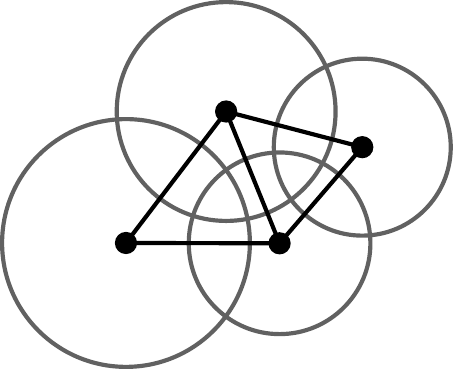} 
        \caption{An intersection graph and its circles where the outer face is adjacent to four vertices.}\label{fig:5b}
    \end{subfigure}
    \caption{}
        \label{fig:tight}
\end{figure}
 
 As we get a contradiction in both cases it follows that the outer face of the intersection graph of an arrangement of at least $n\geq5$ circles is adjacent to at least $5$ vertices.
\end{proof}

\section{Omitted material from Section~\ref{sec:uppernest}}\label{app:omitted}

Before we prove \autoref{8edges} we prove the following lemma.

\begin{lemma} \label{7erGraph}
 Let $G$ be a graph with at most seven vertices and no induced $C_3$ or $C_4$. Then one of the following holds:
 \begin{enumerate}
  \item[(i)]  $G$ has a vertex of degree at most $1$. 
  \item[(ii)]  $G$ is a $C_7$, $C_6$ or a $C_5$. 
  \item[(iii)] $G$ is two $C_5$s glued together at a path with two edges. 
 \end{enumerate}
\end{lemma}
\begin{proof}
If $G$ has at most four vertices, (i) obviously holds.
If $G$ has more than four vertices, we further analyse $G$ as follows.

If (i) does not hold, then $G$ contains a cycle. If the shortest cycle in $G$ has length $7$, (ii) holds. 

If the shortest cycle has length $6$, we have two cases: If $G$ has only $6$ vertices, then (ii) holds. Or if $G$ has $7$ vertices, consider the vertex $v$ not on that cycle of length six. 
If $v$ has at least two neighbors, these neighbors must lie on the $C_6$ and thus yield a shortcut. This is not possible since the $C_6$ is the shortest cycle in $G$. Hence the vertex $v$ is incident to at most one edge and (i) holds.

If the shortest cycle has length $5$, we have three cases. First, if $G$ has only $5$ vertices, then (ii) holds. Second, if $G$ has $6$ vertices, then the sixth vertex is incident to at most one edge by the same argument used in the case of $7$ vertices and a $C_6$ and (i) holds.
Thirdly, if $G$ has $7$ vertices and if (i) and (ii) do not hold, consider the two vertices $v,w$ not on the $C_5$. Since (i) does not hold $v,w$ have degree at least two. If either one of them has two neighbors on the $C_5$, it yields a shortcut and we would have a shorter circle. This contradicts the assumption that the $C_5$ is the shortest cycle in $G$, so both have just one neighbor on the cycle. Since $u$ and $v$ have degree of at least $2$ they are connected by an edge. Let $v'$ and $w'$ be the neighbors of $v$ and $w$ on the five-cycle, respectively. The length of the shortest path from $v'$ to $w'$ on the $C_5$ is either one or two. If it is one, then the path $v',v,w,w',v'$ would be a $C_4$, a contradiction. So the length of the shortest path is two and $G$ is two $C_5$ glued together at the path from $v'$ to $w'$. Thus, (iii) holds.
\end{proof}

\eightedges*
\begin{proof}
We delete vertices with degree~$1$ until no such vertex remains. Let $G'$ be the remaining graph.
Since $G'$ fulfils the condition (ii) or (iii) of \autoref{7erGraph} it can be made acyclic
by deleting at most two edges. Hence, also $G$ can be made acyclic by deleting at most two edges.
It follows that that $G$ has at most $(n-1)+2$ edges, if it has $n$ vertices. Thus, $n\le 7$
implies that $G$ has at most 8 vertices.
\end{proof}

In the remainder of this part we analyse arrangements with a small number of inner black circles 
to prove the upper bound of $\left(4+\frac{5}{11}\right)\cdot n$ edges in the intersection graph of an arrangement of $n$ circles.
We start with a slightly stronger statement of \autoref{4m}. 

\begin{corollary}\label{lem:3inner}
 In the intersection graph of every arrangement of orthogonal circles we can find a nonempty subset $V_C$ that is incident to at most $4n+5$ edges where $n=|V_C|$. Moreover, if at most three vertices 
 corresponding to inner black circles of $V_C$ are incident to two green edges each, then $V_C$ is incident to at most $4n$ edges.
\end{corollary}
\begin{proof}
	We can follow the proof of \autoref{4mi} and use the fact that for the \enquote{Moreover}-part
	we have $i_2\le 3$.
\end{proof}

Before proving the desired bound in \ref{lem:4+m} we provide some necessary lemmas.

\begin{lemma}\label{lem:8around}
Let $C$ be the red circle in an orthogonal circle arrangement and let $S_C$ be the set of the 
black circles properly nested in $C$. 
 If an inner black circle of $S_C$ is intersected by two green circles, 
 then $S_C$ contains at least $8$ boundary black circles. 
\end{lemma}
\begin{proof}
Let $D$ be an inner black circle that is intersected by the two green circles $E$ and $F$. Since  both $E$ and $F$ intersect $D$ and $C$, they must intersect each other. Otherwise the corresponding vertices of $E,D,F$ and $C$ would induce a $C_4$, which would violate \autoref{K4C4}. Thus, $E$ and $F$ intersect and according to \autoref{intersection} $D$ contains one of the intersection points.
If $D$ is an inner circle, then its corresponding vertex $v_D$ in the (embedded) 
intersection graph $G_C$ of $S_C$ is not adjacent to the outer face. So $v_D$ lies on the interior 
of a cycle of vertices corresponding to boundary black circles. 
Let $A=(a_1,...,a_x)$ be the such a cycle that is inclusion-minimal. 

Every green circle intersecting $D$ must also intersect at least two circles corresponding to vertices in $A$. So the vertices $v_E$ and $v_F$ corresponding to the green circles $E$ and $F$ must each have at least two neighbors in $A$. According to \autoref{justOne} and since the black circles are not nested, there can only be one black circle intersecting both $E$ and $F$. This circle is the circle $D$. Thus, the vertices $v_E$ and $v_F$ cannot have a common neighbor.

Let $a_i,a_j$ be the neighbors of $E$ and $a_k,a_l$ be the neighbors $F$. If two of the neighbors are adjacent, say $a_i$ and
$a_k$ then $a_i,v_E,a_k,v_F$ induce a $C_4$, which would contradict \autoref{K4C4}. Thus, $A$ muss contain at least 8 vertices.
\end{proof}

As an immediate consequence from the last lemma we get.
\begin{lemma}\label{lem:m0} 
Let $C$ be the red circle in an orthogonal circle arrangement and let $S_C$ be the set of the circles properly nested in $C$. If $|S_C|\leq 11$, then $S_C$ contains at most three inner black circles that each are intersected by two green circles.
\end{lemma}
\begin{proof}
 Assume for contradiction that $|S_C| \leq 11$ and $S_C$ contains four inner circles that are intersected by two green circles each. It follows that $S_C$ contains at most seven boundary black circles. 
 This is a contradiction to \autoref{lem:8around}.
\end{proof}

\begin{lemma}\label{lem:atMost18}
 In the intersection graph of an arrangement of orthogonal circles we can find a nonempty subset $V_C$ that is incident to at most $4n+5$ edges, where $n=|V_C|$. Moreover, if $n\leq 11$, then the subset is incident to at most $4n$ edges.
\end{lemma}
\begin{proof}
According to \autoref{lem:3inner} there is a non empty subset $V_C$ in the intersection graph of every orthogonal circle arrangement that is incident to at most $4n+5$ edges, where $n=|V_C|$.
However, according to \autoref{lem:m0}, if $n<11$ then there are at most $3$ inner black 
circles that are orthogonal to two green circles. In this case, due to \autoref{lem:3inner}, $V_C$ is incident to at most $4n$ edges.

\end{proof}

\four*

\begin{proof}
According to \autoref{lem:atMost18} there is a nonempty subset $V_C$ in the intersection graph of every orthogonal circle arrangement that is incident to at most $4n+5$ edges, where $n=|V_C|$. Moreover, if $n\leq 11$, then the subset is incident to at most $4n$ edges.

If $n\leq 11$, then $V_C$ is incident to at most $4n<\left(4+\frac{5}{11}\right) n$ edges.

If $n>11$, then $V_C$ is incident to at most $4n+5=\left(4+\frac{5}{n}\right) n\leq\left(4+\frac{5}{11}\right)n$ edges.
\end{proof}

\lowerbound*

\begin{proof}
We set $x=\lfloor \sqrt{ n} \rfloor$ and $a= \lceil \sqrt{n} \rceil -2$.
Note that for any positive real number
$t$ we have $\lfloor t \rfloor (\lceil t \rceil -1) < t^2 < (\lfloor t \rfloor +1) \lceil t \rceil $.
Hence, for our choice of parameters the arrangement $\mathcal B_{x,a}$ has by \autoref{count} 

$$ x\cdot (a+1) = \lfloor \sqrt{ n} \rfloor \left( \lceil \sqrt{ n} \rceil -1\right)  < n$$
vertices. 
We introduce additional independent circles so that $B_{x,a}$ has exactly $n$ vertices.
Also by Lemma~\ref{count}
we get that this arrangement has at least

\begin{align*}
	4xa-2a & = 4 (\lfloor \sqrt{ n} \rfloor)( \lceil \sqrt{n} \rceil -2) - 2 \lceil \sqrt{n} \rceil +4 \\
	  & = 4 (\lfloor \sqrt{ n} \rfloor+1 -1) \lceil \sqrt{n} \rceil - 8 \lfloor \sqrt{n} \rfloor- 2 \lceil \sqrt{n} \rceil +4  \\
	  & = 4 (\lfloor \sqrt{ n} \rfloor+1) \lceil \sqrt{n} \rceil  - O(\sqrt{n})  \\
	  & > 4n - O(\sqrt{n})\\
\end{align*}
 many edges and the statement of the lemma follows. Note that for very small $n$ our construction 
 is degenerate, which is however covered by the big-O error term.
%
%
%
\end{proof}

\lowerNonNest*
\begin{proof}
 Let $G$ be the intersection graph of $\mathcal B_{(n-1)/5,5}$ in which we have deleted all but the innermost hub circle and let $m$ its number of edges. The arrangement is nonnested by construction.
Every vertex in $G$ has degree $6$, except the vertices that correspond to the inner most satellite circles, the only hub circle and the vertices adjacent to the outer face of $G$. The $6$ vertices corresponding to inner most satellite circles and the one hub circle have degree $5$. The $5$ vertices on the outer face of $G$ have degree $4$. This gives us $2m = \sum_{v\in V(G)}\deg(v)=6n-6-5\cdot2=6n-16$. Thus, the intersection graph has $n$ vertices and $3n-8$ edges. 
\end{proof}

\triangular*
\begin{proof}
 We construct the arrangement $\mathcal A$ by taking the arrangement $\mathcal B_{x,a}$ described above and drawing a small circle over every intersection point, such that the small circle only intersects the two circles corresponding to the intersection point. That this is always possible can be seen by taking an inversion in a circle with center on the intersection point. The circles corresponding to the intersection point turn into straight lines. We can now draw a circle around their intersection point that does not intersect any other circles. Reversing the inversion gives us the small circle.
 Note that the inversion is a Möbius transformation and therefore conformal.
 Each of those new circles induces four triangular cells.
 

%
%

  According to \autoref{count} the intersection graph of $\mathcal B_{x,a}$ has $4xa-2a$ edges. Thus, $\mathcal B_{x,a}$ has $8xa-4a$ intersection points. The arrangement $\mathcal A$ has 
 therefore $n=x \cdot (a+1)+8xa-4a$ circles and at least $4\cdot (8xa-4a)$ triangular cells. 
 
 We set $x=2/3\left(\lceil\sqrt{n}\rceil-1\right)$ and $a=1/6\lfloor \sqrt{n} \rfloor$. Note that for any positive real number
$t$ we have $\lfloor t \rfloor (\lceil t \rceil -1) < t^2 < (\lfloor t \rfloor +1) \lceil t \rceil $. Hence, for our choice of parameters the arrangement $A$ has 
\begin{align*}
 9xa-4a+x &= 9\cdot \frac{2}{3}\left(\lceil\sqrt{n}\rceil-1\right) \cdot \frac{1}{6}\lfloor \sqrt{n} \rfloor - 4 \cdot \frac{1}{6}\lfloor \sqrt{n} \rfloor + \frac{2}{3}\left(\lceil\sqrt{n}\rceil-1\right) \\
 &=\left(\lceil \sqrt{n}\rceil - 1 \right)\lfloor \sqrt{n} \rfloor -\frac{2}{3}\left(\lfloor\sqrt n \rfloor - \lceil \sqrt n \rceil +1\right) \\
 &\leq \left(\lceil \sqrt{n}\rceil - 1 \right)\lfloor \sqrt{n} \rfloor < n
\end{align*}
circles. 
We fill $\mathcal A$ up such that it has exactly $n$ circles. So this arrangement has at least
\begin{align*}
 32xa-16a &= 32 \cdot \frac{2}{3}\left(\lceil\sqrt{n}\rceil-1\right) \cdot \frac{1}{6}\lfloor \sqrt{n} \rfloor - 16 \cdot \frac{1}{6}\lfloor \sqrt{n} \rfloor \\
 &=\frac{32}{9} \left(\lceil\sqrt{n}\rceil-1\right) \lfloor \sqrt{n} \rfloor - \frac{8}{3} \lfloor \sqrt{n} \rfloor \\
 &=\frac{32}{9} \lceil\sqrt{n}\rceil \left(\lfloor \sqrt{n} \rfloor +1-1\right) -\frac{32}{9}\lfloor \sqrt{n} \rfloor - \frac{8}{3} \lfloor \sqrt{n} \rfloor \\
 &=\frac{32}{9} \lceil\sqrt{n}\rceil \left(\lfloor \sqrt{n} \rfloor +1\right) - \frac{32}{9}\lceil \sqrt{n} \rceil -\frac{32}{9}\lfloor \sqrt{n} \rfloor - \frac{8}{3} \lfloor \sqrt{n} \rfloor \\
 &>\frac{32}{9} n - O(\sqrt n)
\end{align*}
many triangular cells and the statement of the lemma follows. 
\end{proof}

\section{Acute nonnested circle arrangements.}
\label{app:acute}
We show in this part that Theorem~\ref{thm:main} also holds in a more 
general setting, namely, the embedded intersection graph is also noncrossing
if the intersections of circles is at an angle of at most $\pi/2$. We call circle arrangements
with this property \emph{acute (nonnested) circle arrangements}. The proof of Theorem~\ref{thm:main}
relies on the combination of Lemma~\ref{lem1}--\ref{lem3}. We prove now these three
lemmas for acute nonnested circle arrangements.

%
%

\begin{lemma}\label{lem1a}
 In an acute nonnested circle arrangement the center of a circle $A$ is not contained in a circle except in $A$.
\end{lemma}

\begin{proof}
 Let $A$ and $B$ be two nonnested circles with centers $C_A$ and $C_B$ and radii $r_A$ and $r_B$, respectively. Assume that $C_A$ lies inside $B$.  Obviously, $A$ and $B$ intersect in some angle $\alpha$, since otherwise the circles are nested.
 It holds that $|C_AC_B|^2=r_A^2+r_B^2 -2r_Ar_B  
 \cos(\pi-\alpha)$. Further, since $C_A$ is in $B$, we have $|C_AC_B|<r_B$ and thus $|C_AC_B|^2<r_B^2$.  We get that  $r_A^2+r_B^2 -2r_Ar_B  \cos(\pi-\alpha)< r_B^2$. Since $0\leq \alpha \leq \frac{\pi}{2}$ it follows that $0 \geq \cos(\pi-\alpha)\geq -1$ and thus $r_A^2+r_B^2 < r_B^2$. Since $r_A,r_B \in \mathbb R$ this is a contradiction.
 \end{proof}

\begin{lemma}\label{lem2a}
 In an acute nonnested circle arrangement for every pair of circles $A$ and $B$ and every point $p$ on $A$ it holds that $B$ intersects $C_Ap$ in at most one point.
\end{lemma}

\begin{proof}Let $A$ and $B$ are two circles with centers $C_A$ and $C_B$ and radii $r_A$ and $r_B$, respectively. Assume for a contradiction that $p$ is a point on $A$ such that $B$ intersects $C_Ap$ twice. We call these intersection points $q$ and $s$ with $|C_Aq|<|C_As|<r_A$. The midpoint between $q$ and $s$ is denoted by $t$. 

By \autoref{lem1a} the center $C_B$ of the circle $B$ has to be outside of the circle $A$. So for the circle $B$ to have any point in the inside of the circle $A$, the circle $B$ has to intersect the circle $A$ in a point $u$.     
Since the circles $A$ and $B$ intersect in an angle $\alpha\le \pi/2$, we have an angle $\pi-\alpha$ at $u$ between $C_Au$ and $C_Bu$. And since $sq$ is a chord of the circle $B$,  $sqC_B$ spans an isosceles triangle with height $C_Bt$. Thus, we have a right angle at $t$ between $C_Ap$ and $C_Bt$.
It follows that
  \begin{gather*}
   |C_AC_B|^2 =r_A^2 + r_B^2 -2r_ar_b \cos(\pi-\alpha),\\
      |C_AC_B|^2 =|C_At|^2+|C_Bt|^2 \qquad \text{and} \qquad
   r_B^2 =\left(\frac{|qs|}{2}\right)^2+|C_Bt|^2.
  \end{gather*}
We obtain
 \begin{align*}
  &   &|C_AC_B|^2 &= r_A^2 + r_B^2 -2r_ar_b cos(\pi-\alpha) \\
   &\Leftrightarrow& |C_At|^2+|C_Bt|^2 &= r_A^2 + -2r_ar_b \cos(\pi-\alpha) + \left(\frac{|qs|}{2}\right)^2+|C_Bt|^2  \\
   &\Leftrightarrow& |C_At|^2 &= r_A^2 + \left(\frac{|qs|}{2}\right)^2 -2r_ar_b \cos(\pi-\alpha)\\
 \end{align*}
 
Since $0\leq \alpha \leq \frac{\pi}{2}$ it follows that $0 \geq \cos(\pi-\alpha)\geq -1$.
And thus, $ |C_At| > r_A$.
We see that $t$ lies outside of the circle $A$ and not on $C_Ap$, this is a contradiction. So there is no circle $B$ that intersects the the line segment $C_Ap$ twice.
\end{proof}

\begin{lemma}\label{lem3a}
 In an acute nonnested circle arrangement for every intersecting pair of circles $A$ and $B$ there is no third circle $D$ that intersects the line segment between the centers $C_A$ and $C_B$.
\end{lemma}

\begin{proof}
Let $A$, $B$ and $D$ be three circles with centers $C_A$, $C_B$ and $C_D$ and radii $r_A$, $r_B$ and $r_D$, respectively. The circles $A$ and $B$ intersect. Assume for
a contradiction that the circle $D$ intersects the line segment between $C_A$ and $C_B$.

A circle can intersect a line once or twice. If the circle $D$ intersects the line segment between $C_A$ and $C_B$ only once, either $C_A$  or $C_B$ would be inside $D$; a contradiction to \autoref{lem1a}. Thus, $D$ has to intersect the line segment $C_AC_B$ twice.
We denote these intersection points $q$ and $s$ with $|C_Aq|<|C_As|<|C_AC_B|$ and the midpoint between $q$ and $s$ by $t$. 
Due to \autoref{lem2a} $q$ and $s$ cannot lie in the same circle, so one lies 
in $A$ and the other in $B$. Thus, $D$ intersects both $A$ and $B$.
By \autoref{lem1a} the center $C_D$ of $D$ has to be outside of the circles $A$ and $B$. So for the circle $D$ to have a point in the inside of the circles $A$ and $B$, the circle $D$ has to intersect the circles $A$  (in some point $u_A$) and $B$ (in some point $u_B$). This is the same situation as in Lemma~\ref{lem3}, which is depicted in \autoref{fig:lem3}.

Since $sq$ is a chord of the circle $D$,  $sqC_D$ spans an isosceles triangle with height $C_Dt$. Thus, we have a right angle at $t$ between $C_AC_B$ and $C_Dt$.
%
%
%
According to the law of cosines we obtain for $0\le \alpha',\alpha''\le \pi/2$:

\begin{gather*}
    |C_AC_D|^2 =r_A^2+r_D^2- 2r_Ar_D \cos(\pi-\alpha'), \quad |C_AC_D|^2 = |C_At|^2+|C_Dt|^2, 
 \\
    |C_BC_D|^2 =r_B^2+r_D^2-2r_Br_D \cos(\pi-\alpha''), \quad |C_BC_D|^2 = |C_Bt|^2+|C_Dt|^2,\\
        r_D^2 =\left(\frac{|qs|}{2}\right)^2+|C_Dt|^2  
  \end{gather*}

Combining these equations yields

 \begin{align*}
	|C_At|^2 & = |C_At|^2 + |C_Dt|^2 - |C_Dt|^2 \\
	&= |C_AC_D|^2 - |C_Dt|^2 =r_A^2+r_D^2 - 2r_Ar_D \cos(\pi-\alpha')- |C_Dt|^2\\
	& = r_A^2 +\left(\frac{|qs|}{2}\right)^2 - 2r_Ar_D \cos(\pi-\alpha'). 
\end{align*}
Since $0\leq \alpha' \leq \frac{\pi}{2}$ it follows that $0 \geq \cos(\pi-\alpha')\geq -1$
and therefore $|C_At| > r_A$. By a symmetric argument we see also that $|C_Bt| > r_B$.
We get $|C_At|+|C_Bt| > r_A + r_B$, which is a contradiction.
\end{proof}

We can now combine our observations to prove the following result. The proof is identical 
to the proof of Theorem~\ref{thm:main} with the strengthened lemmas. We repeat it 
here for completeness.

\begin{theorem}\label{thm:angles}
 The embedded intersection graph of an acute nonnested circle arrangement is noncrossing.
\end{theorem}

\begin{proof}
Suppose for contradiction that the embedded intersection graph has two edges $C_AC_B$ and $C_CC_D$ that cross in the point $h$. 
This means we have two pairs of intersecting circles $A,B$ and $C,D$, with corresponding centers $C_A,C_B,C_C,C_D$.
Note that  $C_AC_B$ is contained in the union of the convex hulls of $A$ and $B$. 
Hence, $h$ has to lie in at least one of the circles $A$ or $B$. By the same reasoning
$h$ also has to lie in at least one of the circles $C$ or $D$. Without loss of generality
we can assume that $h$ lies in $C$.
By \autoref{lem1a} the circle $C$ cannot enclose $C_AC_B$ completely, thus it has to intersect the line segment $C_AC_B$. This, however, contradicts \autoref{lem3a}.
\end{proof}

An immediate consequence of Theorem~\ref{thm:angles} and Euler's formula is the following
corollary. 

\begin{corollary}
	The embedded intersection graph of $n$ acute nonnested circles has at most $3n-6$ edges. 
\end{corollary}

Note that the slightly stronger statement of~\autoref{cor:1} does not carry over since
Lemma~\ref{K4C4} does not hold for intersections with acute angles. In particular, the
$K_4$ has a contact representation by touching disks by the circle packing theorem.

\section{Explicit construction of the lower bound examples}
\label{app:lower}
We give in this section a formal construction of the arrangement $\mathcal B_{x,a}$
including a proof of its orthogonality.

We define the following constants and points in the plane for later reference.
\begin{itemize}[noitemsep]
 \item $\alpha = \frac{\sqrt{\cos{\left( \frac{2 \ensuremath{\pi} }{a}\right) }-\cos{\left( \frac{4 \ensuremath{\pi} }{a}\right) }}+\sqrt{2} \cos{\left( \frac{\ensuremath{\pi} }{a}\right) }}{\sqrt{2} \cos{\left( \frac{2 \ensuremath{\pi} }{a}\right) }}$
 \item $d_i=  \frac{\alpha^{i-1}}{\sqrt{2}\cdot \sin\left(\frac{\pi}{a}\right)}$
 \item $s_i=\alpha^{i-1} $
 \item $C_{i,j}= \begin{cases}
    \left(d_i,\frac{2\pi \cdot j}{a}\right), & \text{if } i \text{ is even} \\
    \left(d_i,\frac{2\pi \cdot j}{a}+\frac{\pi}{a}\right), & \text{else}
    \end{cases}$
 \item $h_i=s_i \cdot \sqrt{\frac{1}{{2\cdot \sin\left(\frac{\pi}{a}\right)^2}}-1}$    
\end{itemize}

The construction is guided by a set of $x$ concentric circles $O_i$ centered at the origin of radii 
$d_i$ for $1\leq i \leq x$; we refer to these circles as \emph{orbits}. These circles 
 will not be part of the final arrangement. On each of the orbit  circles $O_i$ 
 we place the centers of $a$ circles $S_{i,j}$ for $1\leq j \leq a$ with radius $s_i$ 
 such that the centers are equidistant; we refer to these as \emph{satellite circles}. 
The center of a circle $S_{i,j}$ for $1\leq i \leq x$, $1\leq j \leq a$ is $C_{i,j}$
As last step we add $x$ concentric circles $H_i$ with center on the origin and 
radii $h_i$ for $1\leq i \leq x$; we refer to these circles as \emph{hub circles}. 
As before we call a hub circle together with the satellite circles it intersects
a \emph{wheel of circles}. 
The satellite circles and the hub circles form the arrangement $\mathcal B_{x,a}$.

We prove now that this arrangement is an orthogonal circle arrangement.

\begin{lemma}
 The arrangement $\mathcal B_{x,a}$ is an orthogonal circle arrangement.
\end{lemma}
\begin{proof}
 We prove the lemma by checking the following conditions.
 \begin{enumerate}
 	\item {\bf In every wheel of circles the satellite circles intersect each other orthogonally.}  	
 	Let us first concentrate on the innermost wheel of circles.
Consider the the equilateral $a$-gon formed by the centers of the satellite circles $C_{1,1},...,C_{1,a}$. An edge length of the $a$-gon is (here between $C_{1,j}$ and $C_{1,j+1}$

\begin{align*}
D &=\sqrt{2\cdot d_1^2 - 2 \cdot d_1^2 \cdot \cos\left(\frac{2\pi \cdot (j+1)}{a}-\frac{2\pi \cdot j}{a}\right)} \\
 &=\sqrt{2\cdot \left(\frac{1}{\sqrt{2}\cdot \sin\left(\frac{\pi}{a}\right)}\right)^2 \left(1- \cos\left(\frac{2\pi}{a}\right)\right)} \\
 &= \sqrt{\frac{1- \cos\left(\frac{2\pi}{a}\right)}{\sin\left(\frac{\pi}{a}\right)^2} } \\
 &= \sqrt 2
\end{align*}

Since $s_1=1$, the radii of the satellite circles is $1$. Thus every two neighboring circles intersect orthogonally. Note that all other wheel of circles are just scaled up copies of the innermost
wheel of circles. Hence, also here the intersections are orthogonally.
\item{\bf In a wheel of circles the hub circle intersects every satellite circle orthogonally.}
Due to rotational symmetry it suffices to prove the statement for one satellite circle. 
We consider now the intersection point of a satellite circle with a hub circle. The radius
of the hub circle is $h_i$ and the radius of the satellite circles is $s_i$. It holds that
$$ h_i^2 + s_i^2 =
\left(s_i \cdot \sqrt{\frac{1}{{2\cdot \sin\left(\frac{\pi}{a}\right)^2}}-1}\right)^2 + s_i^2
= \left(\frac{\alpha^{i-1}}{\sqrt{2}\cdot \sin\left(\frac{\pi}{a}\right)}\right)^2=d_i^2.$$
Since $d_i$ is the distance between between the center of the satellite circle and the center
of the hub circle (origin) we get by \autoref{distance} that both circles intersect orthogonally.
 	
\item {\bf The satellite circles of two adjacent wheels of circles intersect orthogonally. }
Due to symmetry it suffices to check the condition between the circles with center 
$C_{1,1}$ and $C_{2,1}$. 
%
%
%
It holds that 
$$\lVert C_{1,1},C_{2,1}\rVert_2^2= {d_2}^{2}-2 \cos{\left( \frac{\ensuremath{\pi} }{a}\right) } d_1\cdot d_2+{d_1}^{2} = \frac{{{\alpha}^{2}}-2 \cos{\left( \frac{\ensuremath{\pi} }{a}\right) } \alpha+1}{2 {{\sin{\left( \frac{\ensuremath{\pi} }{a}\right) }}^{2}}}
$$

If we can show that $\lVert C_{1,1},C_{2,1}\rVert_2^2={s_1}^2+{s_2}^2=1+\alpha^2$, then by 
\autoref{distance} the orthogonality is proven. With computer algebra software we have checked that
indeed
\[{{\alpha}^{2}}+1=\frac{{{\alpha}^{2}}-2 \cos{\left( \frac{\ensuremath{\pi} }{a}\right) } \alpha+1}{2 {{\sin{\left( \frac{\ensuremath{\pi} }{a}\right) }}^{2}}}.\]
\item {\bf No hub circle intersects any circle that is not part of its wheel of circles.} 
Let $H$ be a hub circle and let $W_1$ and $W_2$ be two neighboring circles from the same wheel of circles. Further
 let $C$ be a circle from an adjacent wheel of circles that intersect $W_1$ and $W_2$.
 We invert all four circles in a circle with center on one of the intersection points of $W_1$ and $W_2$. This turns $W_1$ and $W_2$ into straight lines that intersect orthogonally and the circles $H$ and $C$ into circles that intersect these straight lines orthogonally . 
 The inversions of $H$ and $C$ are nested and thus do not intersect. Thus, also $H$ and $C$ do not intersect.
 \end{enumerate}
 Since these four conditions hold, all intersection points in the arrangement belong to 
 circles that intersect orthogonally.	
\end{proof}

\end{document}